\documentclass[12pt,reqno]{amsart}
\usepackage{dsfont, amssymb,amsmath,amscd,latexsym, amsthm, amsxtra,amsfonts}
\usepackage[all]{xy}
\usepackage{graphicx}
\usepackage{float}
\usepackage{txfonts}
\usepackage[active]{srcltx}
\def\R{{\bf R}}

\newtheorem{theorem}{Theorem}[section]

\newtheorem{example}[theorem]{Example}

\newtheorem{Them}{Theorem}[section]

\newtheorem{Remark}{Remark}[section]
\newtheorem{Lemma}{Lemma}[section]

\begin{document}
\makeatletter
\def\@setauthors{%
\begingroup
\def\thanks{\protect\thanks@warning}%
\trivlist \centering\footnotesize \@topsep30\p@\relax
\advance\@topsep by -\baselineskip
\item\relax
\author@andify\authors
\def\\{\protect\linebreak}%
{\authors}%
\ifx\@empty\contribs \else ,\penalty-3 \space \@setcontribs
\@closetoccontribs \fi
\endtrivlist
\endgroup }
\makeatother
 \baselineskip 18pt
\title[{\small Optimal risk  policy with constraints}]
 { { optimal dividend and investing control of a
 insurance company  with higher solvency constraints }}
 \author[{ \bf Zongxia Liang and Jianping Huang      } ]
{ Zongxia  Liang \\ Department of Mathematical Sciences, Tsinghua
University, Beijing 100084, China. Email:
zliang@math.tsinghua.edu.cn \\
Jianping  Huang \\ Department of Mathematical Sciences, Tsinghua
University, Beijing 100084, China. Email: jp.huang08@gamil.com  }
 \noindent
\begin{abstract}
This paper considers  optimal control problem of a large insurance
company under a fixed  insolvency probability. The company controls
proportional reinsurance rate, dividend pay-outs and investing
process to maximize the expected present value of the dividend
pay-outs until the time of bankruptcy. This paper aims at describing
the optimal return function as well as the optimal policy. As a
by-product, the paper theoretically sets a risk-based capital
standard to ensure the capital requirement of can cover the total
risk.
\end{abstract}
 \maketitle
 \setcounter{equation}{0}
\section{{\bf Introduction}}
\vskip 15pt\noindent
 In this paper we consider  optimal control problem of a
 large insurance company in which the dividend pay-outs, investing
 process and the risk exposure are controlled by management.
 The investing  process in a financial market may contain an element of
 risk, so it will impact security and solvency of the company
 (see Theorem \ref{theorem41} below).
Moreover, the company has a minimal reserve  as its guarantee fund
to protect insureds and attract sufficient number of policy holders.
We assume that the company can only reduce its risk exposure by
proportional reinsurance policy for simplicity. The objective of the
company is to find a policy, consisting of risk control and dividend
payment scheme, which maximizes the expected total discounted
dividend pay-outs until the time of bankruptcy. This is a mixed
regular-singular control problem on diffusion model which
 has been a renewed interest recently, e.g.He and Liang\cite{he and
Liang} and references therein, H{\o}jgaard and Taksar
\cite{s104,s103, s1012}, Harrison and Taksar \cite{s10101}, Paulsen
and Gjessing \cite{s1015},  Radner and
 Shepp \cite{s1018}. Optimizing dividend pay-outs is a
 classical problem in actuarial mathematics, on which earlier work
 is given in e.g. Borch\cite{s1010,s109} and
Gerber\cite{s1011}. We notice that  some of these papers  seem not
to  take security and solvency into consideration
 and so the results therein may not be commonly used in practice
 because the insurance business
  is a business {\sl affected with a public interest},
  and insureds and policy-holders should
  be protected against insurer insolvencies
  (see Williams and Heins\cite{law}(1985), Riegel and
  Miller\cite{law2}(1963), and
   Welson and Taylor\cite{law3}(1959)).
   The policy,   making the company go
bankrupt before termination of contract between insurer and policy
holders or the policy of low solvency(see  \cite{law1}), is not the
best way and should be prohibited even though it can win the highest
profit. Therefore, {\sl one of our  motivations } is to consider
optimal control problem of a large insurance company under higher
 solvency and security,  and to find the best equilibrium
policy between making  profit and improving security.
 \vskip 15pt \noindent
Unfortunately, there are very few results concerning on  optimal
control problem of a large insurance company  based on higher
solvency and security. Paulsen\cite{pau}  studied this
kind of optimal controls for diffusion model via properties of
return function, some of our results  somewhat like that of the
\cite{pau}, but both  approaches used  are very different. He, Hou
and Liang\cite{HHL} investigated the optimal control problem for
linear Brownian model. However, we find that the case treated in
the \cite{HHL} is a trivial case, that is, the company of the
 model in the \cite{HHL} will never go to bankruptcy, it is an ideal
 model in concept, and it indeed does not exist in
 reality(see Theorem \ref{theorem42} below). Because
 probability of bankruptcy for the model treated in the present paper
  is very large
 (see Theorem \ref{theorem41} below),
  our results can not be directly deduced  from  the \cite{HHL}.
  Therefore, to solve these the problems we need to use initiated idea from
    the \cite{HHL}, stochastic analysis and PDE method  to
  establish a complete setting for  further discussing
  optimal control problem  of a large insurance company under
 higher solvency and  security in which the dividend pay-outs, investing
 process and the risk exposure are controlled by management.
 This is {\sl   anther one of  our  motivations.}
 This paper is the first systematic presentation of the topic, and
 the approach here is  rather general, so we
  anticipate that it can deal with
other models. We aim at deriving  the optimal return function, the
optimal retention rate and dividend payout level. The main result of
this paper will be presented in section 3 below. As a by-product,
the paper theoretically sets a risk-based capital standard to ensure
the capital requirement of can cover the total given risk. Moreover,
we also discuss how the risk and minimum reserve requirement affect
the optimal reactions of the insurance company by the implicit types
of solutions and how the optimal retention ratio and dividend payout
level are affected by the changes in the minimum reserve requirement
and risk faced by the insurance company.

 \vskip 15pt \noindent
 The paper is organized as follows: In next section 2 we
 establish a stochastic control model of a large insurance
company. In section 3 we present  main result of this paper and its
economic and financial interpretations, and discuss how the risk and
minimum reserve requirement affects  the optimal retention ratio and
dividend payout level of the insurance company. In section 4 we give
analysis on risk of  stochastic control model treated in the present
paper and study relationships among investment risk, underwriting
risk and the insolvency probability. In section 5 we give some
numerical samples to portray how the risk and minimum reserve
requirement affect dividend payout level of the insurance company.
The proofs of theorems and lemmas which study properties of
probability of bankruptcy and optimal return function will be given
in the appendix. \vskip 10pt\noindent
 \setcounter{equation}{0}
\section{ {\bf Mathematical model}}
\vskip 10pt\noindent To give a mathematical
formulation of the optimization problem
 treated in this paper, let $(\Omega, \mathcal {F}, \{ \mathcal
{F}_{t}\}_{t\geq 0}, \mathbb{P})$ denote a
filtered probability space. For the  {\sl
intuition} of our diffusion model we start
from the classical Cram\'{e}r-Lundberg
model of a reserve(risk) process. In this
model claims arrive according to a Poisson
process $N_t$ with intensity $\nu $ on
$(\Omega, \mathcal {F}, \{ \mathcal
{F}_{t}\}_{t\geq 0}, \mathbb{P})$. The size
of each claim is $U_i$. Random variables
$U_i$ are i.i.d. and are independent of the
Poisson process $N_t$ with finite first and
second moments given by $\mu$ and $
\sigma^2$ respectively. If there is no
reinsurance, dividend pay-outs or
investments, the reserve (risk) process of
insurance company is described by
$$ r_t= r_0 +pt-\sum^{N_t}_{i=1} U_i,$$
where $p$ is the premium rate. If $\eta > 0 $ denotes the {\sl
safety loading}, the $p$ can be calculated via the expected value
principle as
$$ p=(1+\eta ) \nu \mu. $$
In a case where the insurance company shares risk with the
reinsurance, the sizes of the claims held by the insurer become
$U^{(a)}_i $, where $a$ is a (fixed) retention level.  For
proportional reinsurance, $a$ denotes the fraction of the claim
covered by cedent. Consider the case of {\sl  cheap reinsurance} for
which the reinsuring company uses the same safety loading as the
cedent, the reserve process of the cedent is given by
$$ r^{(a,\eta )}_t=u + p^{(a,\eta )}t - \sum^{N_t}_{i=1} U^{(a)}_i, $$
where $p^{(a,\eta)}=(1+\eta )\nu \mathbb{E}\{U^{(a)}_i\}. $ Then as
$\eta \rightarrow 0  $
\begin{eqnarray}\label{xx1}
 \{\eta r^{(a,\eta )}_{t/\eta^2 }\}_{t\geq
0}\stackrel{D}{\rightarrow} BM(\mu (a)t, \sigma^2(a)t    )
\end{eqnarray}
in $\mathcal{D}[0, \infty)   $ (the space of right continuous
functions with left limits endowed with the skorohod topology),
where
$$ \mu (a)= \nu \mathbb{E}\{  U^{(a)}_i \}, \qquad
 \sigma^2(a)=\nu \mathbb{E}\{  U^{(a)}_i \}^2, $$
and $ BM ( \mu, \sigma^2)$ stands for Brownian motion with the drift
coefficient $\mu $ and diffusion coefficient $\sigma $ on $(\Omega,
\mathcal {F}, \{ \mathcal {F}_{t}\}_{t\geq 0}, \mathbb{P})$. The
passage to the limit works well in the presence of a big
portfolios. We refer the reader for this fact and for the specifies
of the diffusion approximations to Emanuel,Harrison and Taylor\cite{C1}(1975),
 Grandell\cite{C2}(1977), Grandell\cite{C3}(1978), Grandell\cite{C4}(1990),
Harrison\cite{C5}(1985), Iglehart\cite{C6}(1969), and Schmidli\cite{C7}(1994).
\vskip 10pt\noindent Throughout this paper we consider the retention
level to be the control parameter selected at each time $t$ by the
insurance company. We denote this value by $a(t)$. If there is no
dividend pay-outs or investments, in view of (\ref{xx1}), we can
assume that   in our model the reserve process $\{R_t\}$ of the
insurance company is given by
\begin{eqnarray*}\label{f1} dR_t=a(t)\mu dt+a(t)\sigma d W_t^1,
\end{eqnarray*}
where $ U^{(a)}_i =aU_i$, $\mu (a)= a \mathbb{E}\{U_i\}$ and $
\sigma^2 (a)=a^2\sigma^2$. And the reserve  invested in a financial
asset is the price process $ \{P_t\}$ governed by
\begin{eqnarray*}
dP_t=rP_tdt+\sigma_p P_tdW_t^2,
\end{eqnarray*}
where $r>0$, $\sigma_p\geq 0$, $\{W_t^1\}_{t\geq 0}$ and
$\{W_t^2\}_{t\geq 0}$ are  two independent standard  Brownian
motions on  $(\Omega, \mathcal {F}, \{ \mathcal {F}_{t}\}_{t\geq 0},
\mathbb{P})$. The case of $\sigma_p=0$ corresponds to the situation
where only risk free assets, such as bonds or bank accounts are
used for investments.
 \vskip 10pt\noindent
 A  policy $\pi$ is  a pair of non-negative c\`{a}dl\`{a}g $ \mathcal
{F}_{t}$-adapted  processes $\{a_\pi (t),L_t ^\pi\}$, where $a_\pi
(t)$ corresponds to the  risk exposure at time $t$ and $L_t ^ \pi$
corresponds to the  cumulative amount of dividend pay-outs
distributed up to time $t$. A policy $\pi =\{a_\pi (t),L_t ^\pi\} $
is called admissible if $ 0\leq a_\pi (t)\leq 1 $ and $L_t ^\pi$ is
a nonnegative, non-decreasing, right-continuous function. When $\pi
 $ is applied, the resulting reserve process is denoted by $\{ R_t^\pi
 \}$. We assume that  the initial reserve $R^\pi_0$ is  a
 deterministic value  $x$. In view of independence of $W^1$ and $W^2$,
 the dynamics for  $R_t^\pi$ is given by
 \begin{eqnarray}\label{f2}
dR_t^\pi=(a_\pi(t)\mu+rR_t^\pi
)dt+\sqrt{a_\pi^2(t)\sigma^2+\sigma_p^2\cdot (R_t^\pi) ^2}\  d
W_t-dL_t^\pi,
\end{eqnarray}
where $\{W_t\}$ is a standard Brownian motion on $(\Omega, \mathcal
{F}, \{ \mathcal {F}_{t}\}_{t\geq 0}, \mathbb{P})$. Moreover, we
suppose that  the insurance company  has a minimal reserve $m$ as
its guarantee fund to protect insureds and attract sufficient number
of policy holders, that is, the company needs to keep its reserve
above $m$. The company is considered bankrupt as soon as the reserve
falls below $ m $.  We define the time of bankruptcy by $\tau ^\pi
_x =\inf\{t\geq 0: \R_t^ {\pi}\leq m \}$. Obviously, $\tau ^\pi_x $
is an $\mathcal {F}_{t}$ -stopping time. \vskip 10pt\noindent
 We denote by $\Pi$ the set of all admissible
policies. For any $b\geq 0$ , let $\Pi_b=\{\pi\in \Pi : \int _0 ^{
\infty }\mathcal{1}_{\{s:R^\pi(s)<b\}}dL_s^{\pi}=0\} $. Then it is
easy to see  that $\Pi=\Pi_0$ and $b_1>b_2\Rightarrow
\Pi_{b_1}\subset \Pi_{b_2}$. For a given admissible policy $\pi$ we
define the optimal return function $V(x)$ by
\begin{eqnarray}\label{f3}
J(x,\pi)&=&\mathbb{E}\big \{\int_0^{\tau^\pi_x } e^{-ct} dL_t^\pi\big\},\nonumber \\
 V(x,b )&=&\sup_{\pi \in \Pi_b}\{J(x,\pi)\},
\end{eqnarray}
\begin{eqnarray}\label{f4}
 V(x)&=&\sup_{b \in \mathfrak{B}}\{  V(x,b )  \}
\end{eqnarray}
and the optimal policy $\pi^* $ by
\begin{eqnarray}\label{f5}
J(x,\pi^*)= V(x),
\end{eqnarray}
 where
  \begin{eqnarray*}
 \mathfrak{B}:=\big \{b\ :\ \mathbb{P}[\tau_{b}^{\pi_{b}} \leq T] \leq
\varepsilon \ , \  J(x, \pi_b)= V(x,b) \mbox{ and} \ \pi_b \in
\Pi_b\big\},
  \end{eqnarray*}
 $c>0$ is a discount rate, $\tau_b^{\pi_b}$ is
 the time of bankruptcy $\tau_x^{\pi_b} $  when
the initial reserve $x=b$ and the control policy is $\pi_b$.
$1-\varepsilon$ is the standard of security and less than solvency
for given $\varepsilon>0 $. \vskip 15pt\noindent {\sl The main
purpose of this paper is to find the optimal return function $V(x)$
 and the optimal policy $\pi^*$.}  Throughout this paper
we assume that $r\leq c$ in view of   $V(x)=\infty$ for $r>c$(see
H{\o}jgaard and Taksar \cite{s104}). \vskip 5pt\noindent
 \setcounter{equation}{0}
\section{ \bf Main result   }
\vskip 15pt\noindent
 In this section we first introduce an  auxiliary
 Hamilton-Jacobi-Bellman (HJB) equation, then we present
  main result of this paper, finally we give economic and financial
  interpretations of the main result.
\begin{Lemma}\label{lemma1} Let  $h\in C^2 [m,\infty)$ satisfy  the
following HJB equation
\begin{eqnarray}\label{31}
\max_{a\in[0,1]}\big \{\frac{1}{2}[\sigma^2 a^2+\sigma_p^2
x^2]h''(x)+[\mu a +rx]h'(x)-ch(x)\big  \}=0,\ x\geq m \nonumber\\
\end{eqnarray}
with boundary condition $ h(m)=0$. Then\\
(i) $h'(x)>0 $, $\forall x\geq m$.\\
(ii) There exists a unique $b_0>x_0$ such that $h''(b_0)=0$ and
$(x-b_0)h''(x)>0$ for all $x\geq m$ except $b_0$, where $
x_0=\frac{\sigma^2(1-\alpha^\ast)}{\mu}$, $\alpha^\ast $is a
constant in $(0,1)$.
\end{Lemma}
\begin{proof}
 The proof  of this lemma is standard and can be proved
 by the same way as in  the proof of He and
  Liang \cite{he and Liang1}, Shreve, Lehoczk
and Gaver \cite{s44} and Paulsen and Gjessing \cite{s1015}. So we
omit it here.
\end{proof}
\vskip 10pt\noindent Assume  that
 $h(x)$ is a  solution of
(\ref{31}). Define  functions $ F_b(x)$ and $a^\ast(x)$  by
\begin{eqnarray}\label{32}
F_b(x)=\left\{
\begin{array}{l l l}0,&0\leq x<m,\\
\frac
{h(x)}{h'(b)},&m\leq x \leq b,\\
x-b+F_b(b),&x\geq b
\end{array}\right.
\end{eqnarray}
and
\begin{eqnarray}\label{33}
a^\ast(x)=\left\{ \begin{array}{l l l} \lambda x , & 0\leq x \leq
x_0,\\
1, & x\geq x_0
\end{array}\right.
\end{eqnarray}
respectively, where  $ \lambda=\frac{\mu}{\sigma^2(1-\alpha^\ast)}$.
It easily follows that $F_b\in C^2 ([m,\infty)\\\{b\})$.
Now we can present the main result of this paper as follows.
We will give rigorous proof of the main result in the appendix.
\begin{Them}\label{theorem31}
Let  level of risk $ \varepsilon \in (0, 1)$ and time horizon $T$ be
given. \vskip 10pt\noindent
 (i) If $\mathbb{P}[\tau_{b_0}^{\pi_{b_o}^\ast}\leq
T]\leq\varepsilon$ then the optimal return function $V(x)$ is
$F_{b_0}(x)$ defined by (\ref{32}), and
$V(x)=F_{b_0}(x)=V(x,0)=J(x,\pi_{b_o}^\ast) $. The optimal  policy
$\pi_{b_o}^\ast $ is
$\{a^\ast(R^{\pi_{b_o}^\ast}_t),L^{\pi_{b_o}^\ast}_t\}$, where
$\{R^{\pi_{b_o}^\ast}_t, L^{\pi_{b_o}^\ast}_t  \} $ is uniquely
determined  by the following stochastic differential
 equation
\begin{eqnarray}\label{34}
\left\{
\begin{array}{l l l}
dR_t^{\pi_{b_o}^\ast}=(a^\ast(R^{\pi_{b_o}^\ast}_t)\mu+rR_t^{\pi_{b_o}^\ast}
)dt+\sqrt{\big (a^\ast(R^{\pi_{b_o}^\ast}_t)\big
)^2\sigma^2+\sigma_p^2\cdot
(R_t^{\pi_{b_o}^\ast}) ^2}\  d W_t\\
\qquad \qquad \qquad \qquad -dL_t^{\pi_{b_o}^\ast},\\
m\leq R^{\pi_{b_o}^\ast}_t\leq  b_0,\\
\int^{\infty}_0 I_{\{t: R^{\pi_{b_o}^\ast}_t
<b_0\}}(t)dL_t^{\pi_{b_o}^\ast}=0.
\end{array}\right.
\end{eqnarray}
 The solvency of the company is bigger than $1-\varepsilon$.
\vskip 10pt\noindent
 (ii) If $\mathbb{P}[\tau_{b_0}^{\pi_{b_o}^\ast}\leq T]>\varepsilon$
  then
there is a unique optimal dividend $b^\ast(\geq b_0)$ satisfying
$\mathbb{P}[\tau _{b^\ast}^{\pi_{b^*}^\ast}\leq T]= \varepsilon $.
The optimal return  function $V(x)$ is $F_{b^*}(x)$ defined by
(\ref{32}), that is,
\begin{eqnarray}\label{eq4.27}
V(x)=F_{b^*}(x)=\sup_{b\in \mathfrak{B}}\{V(x,b)\},
\end{eqnarray}
where
\begin{eqnarray}
\label{eq38}b^\ast=\min\{b: \mathbb{P}[\tau_{b}^{\pi_{b}} \leq T] =
\varepsilon                \}=\min\{b: b\in \mathfrak{B}\}\in
\mathfrak{B}
\end{eqnarray}
and
\begin{eqnarray*}\label{fe3}
 \mathfrak{B}:=\big \{b: \mathbb{P}[\tau_{b}^{\pi_{b}} \leq T] \leq
\varepsilon, \  J(x, \pi_b)= V(x,b) \mbox{ and} \ \pi_b \in \Pi_b\
\big\}.
 \end{eqnarray*}
Moreover,
\begin{eqnarray}
\label{eq37}V(x)=V(x,b^*)=J(x, \pi^*_{b^*})
\end{eqnarray}
and the optimal policy $\pi_{b^*}^\ast$ is $
 \{a^\ast(R^{\pi_{b^*}^\ast}_t),L^{\pi_{b^*}^\ast}_t\}$, where
$\{R^{\pi_{b^*}^\ast}_t, L^{\pi_{b^*}^\ast}_t  \} $ is uniquely
determined  by the following stochastic differential
 equation
\begin{eqnarray}\label{37}
\left\{
\begin{array}{l l l}
dR_t^{\pi_{b^*}^\ast}=(a^\ast(R^{\pi_{b^*}^\ast}_t)\mu+rR_t^{\pi_{b^*}^\ast}
)dt+\sqrt{\big (a^\ast(R^{\pi_{b^*}^\ast}_t)\big
)^2\sigma^2+\sigma_p^2\cdot
(R_t^{\pi_{b^*}^\ast}) ^2}\  d W_t\\
\qquad \qquad \qquad \qquad -dL_t^{\pi_{b^*}^\ast},\\
m\leq R^{\pi_{b^*}^\ast}_t\leq  b^*,\\
\int^{\infty}_0 I_{\{t: R^{\pi_{b^*}^\ast}_t
<b^*\}}(t)dL_t^{\pi_{b^*}^\ast}=0.
\end{array}\right.
\end{eqnarray}
The solvency of the company is $1-\varepsilon$. \vskip 10pt\noindent
(iii) For any $x\leq b_0$,
 \begin{eqnarray}\label{38}
 \frac{ F_{b^*}(x)  }{F_{b_0}(x)   }=\frac{h'(b^\ast)}{h'(b_0)}<1.
 \end{eqnarray}
\end{Them}
\vskip 1cm\noindent { \bf  Economic and financial explanation of
theorem \ref{theorem31} is as follows:}
 \vskip 10pt\noindent
{\sl  (1) For a given level of risk  and time horizon, if
probability of
 bankruptcy is less  than the level of risk, the optimal
  control problem of
(\ref{f4}) and (\ref{f5}) is the traditional one, the company has
higher solvency, so it will have  good reputation. The solvency
constraints here do not work. This is a trivial case. In view of
Theorem \ref{theorem42} below, the model treated in \cite{HHL} can
be reduced to this trivial case. \vskip 10pt\noindent
 (2) If probability of bankruptcy is
large than the level of risk, the traditional optimal policy will
not meet the standard of security and solvency, the company needs to
find a sub-optimal policy $\pi_{b^*}^\ast $ to improve its solvency.
The sub-optimal reserve process $ R^{\pi_{b^*}^\ast}_t $ is a
diffusion process reflected at $b^*$, the process
$L^{\pi_{b^*}^\ast}_t  $ is the process which ensures the
reflection. The sub-optimal action is to pay out everything in
excess of $b^*$ as dividend and pay no dividend when the reserve is
below $b^*$, and $ a^*(x)$ is the sub-optimal feedback control
function.\vskip 10pt\noindent (3) On the one hand, the inequality
(\ref{38}) states that $\pi_{b^*}^\ast $ will reduce the company's
profit, on the other hand, in view of (\ref{eq38}) and
$\mathbb{P}[\tau _{b^\ast}^{\pi_{b^*}^\ast}\leq T]= \varepsilon $ as
well as lemma \ref{lemma4} below, the cost of improving solvency is
minimal. Therefore the policy $\pi_{b^*}^\ast $ is the best
equilibrium action between making profit and improving solvency. }
\vskip 1cm
\noindent {\bf Effect of  the risk level $\varepsilon$ and
minimum reserve requirement $m$ on  the optimal reaction and
dividend payout level of the insurance company is given  as follows:}
\vskip 10pt\noindent
{\sl (4) We see from the figure \ref{test2m=1,2} below( based on PDE(\ref{4.16})satisfied by solvency probability) that the dividend payout level $b^*$ is an increasing function of minimum reserve requirement $m$. Using comparison theorem for one-dimensional It\^{o} process we know that the reserve process $R_t^{\pi_{b^*}^\ast}$ of the insurance company is  also an  increasing function of $b^*$. Therefore, since the sub-optimal feedback control function $ a^*(x)$ is increasing with respect to $x$, by theorem\ref{theorem31} we conclude that the optimal retention ratio $a^\ast(R^{\pi_{b_o}^\ast}_t)$  increases with $m$, that is, increasing minimum reserve requirement  will improve the optimal retention ratio. However, this increasing action must result in lower  profit because the optimal return function $ V(x,b^*)$ is a decreasing of $b^*$(see Lemma \ref{lemma4}). So the process $L^{\pi_{b^*}^\ast}_t  $ is a decreasing function of $m$ too.
\vskip 10pt\noindent
(5) We see from the figure \ref{test22m=1} below that the dividend payout level $b^*$ is a decreasing function of the risk $\varepsilon$. So, by the same argument as in (4) above, the optimal retention ratio $a^\ast(R^{\pi_{b_o}^\ast}_t)$  decreases with $\varepsilon$, the process $L^{\pi_{b^*}^\ast}_t  $ increases with $\varepsilon$.
 }
 \vskip 10pt  \noindent
(6) We also see from the figure \ref{test3sigma=1,41} below that, for given the risk $\varepsilon$, the dividend payout level $b$ is an increasing function of underwriting risk $\sigma^2$, so it decreases the company's profit.
\begin{Remark}
Because the \cite{HHL} had no continuity of probability of
bankruptcy and actual $b^*  $, the authors of \cite{HHL}  did not
obtain the best equilibrium policy $\pi_{b^*}^\ast $.
\end{Remark}
\begin{Remark}
By (\ref{4.16}) one knows that the equation $\psi(T,m, b^*)=1-\phi(T,m,b^*)=\varepsilon$
can set a risk-based capital standard $( m,b^*)$ to ensure
the capital requirement of can cover the total given risk $\varepsilon$, then establish the optimal return function, the
optimal retention rate and dividend payout level via Theorem \ref{theorem31}.
\end{Remark}
\begin{Remark}
By using the same approach as in \cite{s104} we can show that the
$b^* $ is an increasing function  of  $\sigma_p^2$, so  the company
has possibility of making larger gain from the reinvestments. We
omit the analysis here. We focus on  the effect of investments risk
on probability of bankruptcy for the  topic of this
paper in next section.
\end{Remark}
\vskip 10pt\noindent
 \setcounter{equation}{0}
\section{\bf Analysis on risk of a large insurance company }
 \vskip 10pt\noindent
 The first result of this
section is the following, which states that the company has to find
optimal policy to improve its solvency.
\begin{Them}\label{theorem41}For $b\geq m>0$, let
 $\{R^{\pi_{b}^\ast}_t, L^{\pi_{b}^\ast}_t  \} $ be defined  by the
following SDE( see Lions and Sznitman \cite{s10100})
\begin{eqnarray}\label{c21}
\left\{
\begin{array}{l l l}
dR_t^{\pi_{b}^\ast}=(a^\ast(R^{\pi_{b}^\ast}_t)\mu+rR_t^{\pi_{b}^\ast}
)dt+\sqrt{\big (a^\ast(R^{\pi_{b}^\ast}_t)\big
)^2\sigma^2+\sigma_p^2\cdot
(R_t^{\pi_{b}^\ast}) ^2}\  d W_t\\
\qquad \qquad \qquad \qquad -dL_t^{\pi_{b}^\ast},\\
m\leq R^{\pi_{b}^\ast}_t\leq  b,\\
\int^{\infty}_0 I_{\{t: R^{\pi_{b}^\ast}_t
<b\}}(t)dL_t^{\pi_{b}^\ast}=0,\\
R_0^{\pi_{b}^\ast}=b.
\end{array}\right.
\end{eqnarray}
Then
\begin{eqnarray}\label{R4.1}
\mathbb{P}\{\tau_{b}^{\pi_{b}^\ast}\leq T \} \geq \varepsilon (b, T)
\equiv \frac{4[1-\Phi(\frac{b-m}{\sqrt{\kappa T}})]^2}{\exp\{
\frac{(\lambda \mu +r)^2 T}{\sigma_p^2}\}}>0 ,
\end{eqnarray}
 where $\tau_{b}^{\pi_{b}^\ast}
=\inf\{t:R^{\pi_{b}^\ast}_t\leq m\}$, $k=(\lambda ^2
\sigma^2+\sigma_p^2)m^2 $, $ \lambda=\frac{\mu}{\sigma^2(1-\alpha^\ast)}$.
\end{Them}
\begin{proof} Since $a^*(x)$ is a bounded  Lipschitz continuous function,
the following SDE
\begin{eqnarray*}
dR_t^{(1)}=(a^\ast(R_t^{(1)})\mu+rR_t^{(1)})dt+
\sqrt{{a^\ast}^2(R_t^{(1)})\sigma^2+\sigma_p^2
{R_t^{(1)}}^2}dW_t,R_0^{(1)}=b
\end{eqnarray*}
has a unique solution  $R_t^{(1)}$. Using comparison theorem for
one-dimensional It\^{o} process, we have
\begin{eqnarray}\label{R4.2}
\mathbb{P}\{R_t^{(1)}\geq R^{\pi_{b}^\ast}_t \}=1.
\end{eqnarray}
Let $\mathbb{Q} $ be a measure  on $\mathcal{F}_T $ defined   by
\begin{eqnarray}\label{R4.3}
d\mathbb{Q}(\omega)=M_T(\omega)d\mathbb{P}(\omega),
\end{eqnarray}
where
 \begin{eqnarray*} M_t & = &\exp \Large\{ -\int_0^t
\frac{a^\ast(R_s^{(1)})\mu+rR^{(1)}_s}{\sqrt{{a^\ast}^2(R_s^{(1)})\sigma^2+\sigma_p^2
{R_s^{(1)}}^2}}dW_s\\
& &\qquad  -\frac{1}{2}\int_0^t
\frac{(a^\ast(R_s^{(1)})\mu+rR_s^{(1)})^2}{{a^\ast}^2(R_s^{(1)})\sigma^2+\sigma_p^2
{R_s^{(1)}}^2}ds \Large\}.
\end{eqnarray*}
Since $\{ M_t\}$ is a martingale w.r.t.$\mathcal {F}_t$, we have
$\mathbb{E} \big [ M_T \big ] =1$. Using Girsanov theorem, we know
that $\mathbb{Q}$ is a probability measure on $\mathcal {F}_T$ and
the process $\{R^{(1)}_t\}$ satisfies the following SDE
\begin{eqnarray*}
dR_t^{(1)}=\sqrt{{a^\ast}^2(R_t^{(1)})\sigma^2+\sigma_p^2
{R_t^{(1)}}^2}d\tilde{W}_t,R_0^{(1)}=b,
\end{eqnarray*}
where $\tilde{W}_t$ is a Brownian motion  on $(\Omega, \mathcal {F},
\{ \mathcal {F}_{t}\}_{t\geq 0}, \mathbb{Q})$. \vskip 5pt\noindent
In view of (\ref{R4.2}),  $R_t^{(1)}\geq R^{\pi_{b}^\ast}_t\geq m>0$
for any $t\geq 0$, so we can define  $\rho(t)$ by
\begin{eqnarray*}
\dot{\rho}(t)=\frac{1}{{a^\ast}^2(R_t^{(1)})\sigma^2+\sigma_p^2
{R_t^{(1)}}^2}
\end{eqnarray*}
and define $\hat{R}_t^{(1)}$ by $R_{\rho(t)}^{(1)}$. Then $\rho(t)$
is a strictly increasing function  and
\begin{eqnarray*}
\hat{R}_t^{(1)}=b+\hat{W}_t,
\end{eqnarray*}
where $\hat{W}_t$ is a standard Brownian motion on $(\Omega,
\mathcal {F}, \{ \mathcal {F}_{t}\}_{t\geq 0}, \mathbb{Q})$.
Moreover, for $t\geq 0$
\begin{eqnarray*}\label{R4.5}
\dot{\rho}(t)&=&\frac{1}{{a^\ast}^2(R_t^{(1)})\sigma^2+\sigma_p^2
{R_t^{(1)}}^2}\\&\leq& \frac{1}{(\lambda ^2
\sigma^2+\sigma_p^2)m^2}\\&:=&\frac{1}{\kappa} >0,
\end{eqnarray*}
 so  $\rho(t)\leq \frac{1}{\kappa} t$ and $\rho^{-1}(t) \geq \kappa
t$. As a result
\begin{eqnarray}\label{244}
\mathbb{Q}[\inf\{t:R_t^{(1)}\leq m\}\leq T]&=&\mathbb{Q}[\inf\{t:
\hat{R}_{\rho^{-1}(t)}^{(1)}\leq m\}\leq
T]\nonumber\\&=&\mathbb{Q}[\inf\{\rho(t): b+\hat{W}_t\leq m\}\leq T]\nonumber\\
&=&\mathbb{Q}[\inf \{t: \hat{W}_t\leq m - b \}\leq
\rho^{-1}(T)]\nonumber\\&\geq& \mathbb{Q}[\inf\{t: \hat{W}_t\leq
m-b\}\leq \kappa T]\nonumber\\&=&2[1-\Phi(\frac{b-m}{\sqrt{\kappa
T}})]>0,
\end{eqnarray}
where $\Phi(\cdot)$ is the standard normal distribution function. By
virtue of (\ref{R4.3}),
\begin{eqnarray}\label{245}
\mathbb{Q}[\inf\{t:R_t^{(1)}\leq m\}\leq T]&=&\int_\omega
\mathcal{1}_{[\inf\{t:R_t^{(1)}\leq m\}\leq
T]}d\mathbb{Q}(\omega)\nonumber\\&=&\int_\omega
\mathcal{1}_{[\inf\{t:R_t^{(1)}\leq m\}\leq
T]}M_Td\mathbb{P}(\omega)\nonumber\\&=&\mathbb{E}^{\mathbb{P}}
[M_T\mathcal{1}_{[\inf\{t:R_t^{(1)}\leq m\}\leq
T]}]\nonumber\\
&\leq&\mathbb{E}^{\mathbb{P}}
[M_T^2]^{\frac{1}{2}}\mathbb{P}[\inf\{t:R_t^{(1)}\leq m\}\leq
T]^{\frac{1}{2}}.\nonumber\\
\end{eqnarray}
Substituting  (\ref{244})  and
\begin{eqnarray*}
\mathbb{E}^{\mathbb{P}}[M_T^2]\leq \exp\{ \frac{(\lambda \mu +r)^2
T}{\sigma_p^2}\},
\end{eqnarray*}
into  (\ref{245}), we get
\begin{eqnarray*}
\mathbb{P}[\inf\{t:R_t^{(1)}\leq m\}\leq
T]&\geq&\frac{\mathbb{Q}[\inf\{t:R_t^{(1)}\leq m\}\leq
T]^2}{\mathbb{E}^{\mathbb{P}}[M_T^2]}\\&\geq&\frac{4[1-\Phi(\frac{b-m}{\sqrt{\kappa
T}})]^2}{\exp\{ \frac{(\lambda \mu +r)^2 T}{\sigma_p^2}\}}>0.
\end{eqnarray*}
Thus by (\ref{R4.2})
\begin{eqnarray}
\mathbb{P}[\tau_{b}^{\pi^*_b}\leq T]& \geq
&\mathbb{P}[\inf\{t:R_t^{(1)}\leq m\}\leq T] \\
&\geq & \varepsilon (b,\sigma^2,\sigma^2_p, T)\equiv\frac{4[1-\Phi(\frac{b-m}{\sqrt{\kappa T}})]^2}{\exp\{
\frac{(\lambda \mu +r)^2 T}{\sigma_p^2}\}}>0.\nonumber
\end{eqnarray}
\end{proof}
 \vskip 10pt \noindent
 {\bf The economic interpretation of  theorem
 \ref{theorem41} is
the following.}\vskip 10pt \noindent
 {\sl   (1)\ The lower  boundary $\varepsilon
(b,\sigma^2, \sigma^2_p, T)$ of bankrupt probability for the company
is an increasing function
of $ \sigma^2_p$, thus the  reinvestments will make the company have
larger risk. \vskip 10pt \noindent(2) \ The lower boundary $\varepsilon (b,\sigma^2,\sigma^2_p, T)$
of bankrupt probability for the company is an increasing function of $ m$, so
the minimum reserve requirement $m$ will increase the risk of  the company
goes to bankruptcy.
\vskip 10pt \noindent
(3)\ The lower boundary $\varepsilon (b,\sigma^2,\sigma^2_p, T)$
of bankrupt probability for the company is a decreasing function of $ b$, so the
optimal dividend payout barrier should keep reasonable high so that
the company gets good solvency.
 \vskip 10pt \noindent (4)\ The company does have larger
risk before the contract between insurer and policy holders  goes
into effect (i.e., $0<T $ is less than the time of the contract issue
) because the lower  boundary $\varepsilon (b,\sigma^2,\sigma^2_p, T)$ is
positive for any $T>0$,
 the company has to find an optimal  policy to improve the
ability of the insurer to fulfill its obligation to policy holders.}
\vskip 10pt\noindent
 Now we  prove the second result of this section.
\begin{Them}\label{theorem42}
Let $m=0$ in Theorem \ref{theorem41}. Then for any T and b
$\mathbb{P}[\tau_{b}^{\pi^*_b}\leq T]=0.$
\end{Them}
\begin{proof} Let $\tau _b ^{\pi^*_b}=inf\{t:R^{\pi^*_b}_t=0,R^{\pi^*_b}_0=b\}$, $
\tau_n=inf\{t:R^{\pi^*_b}_t=2^{-2n}x_0\}$, $A=\{\tau^{\pi^*_b}_b
\leq T\}$ and $B_n=\{\tau_n\leq T\}$. Then for any $n>0$  $A\subset
B_n$. As a result,
\begin{eqnarray*}
\mathbb{P}[A] =\mathbb{P}[A\bigcap B_n]\leq \mathbb{P}[A|B_n].
\end{eqnarray*}
Noting that $ \{ R^{\pi^*_b}_t\}$ is a Markov process, we have
\begin{eqnarray*}
\mathbb{P}[A|B_n]&=&\mathbb{P}[\inf_{0\leq t \leq
T}R^{\pi^*_b}_t\leq 0|\tau_n\leq
T]\\&\leq&\mathbb{P}^{2^{-2n}x_0}[\inf_{0\leq t \leq
T}R^{\pi^*_b}_t\leq 0]\\&\leq& \mathbb{P}^{2^{-2n}x_0}[\inf_{0\leq t
\leq T}R^{\pi^*_b}_t\leq 2^{-3n}x_0\ \mbox{or} \sup _{0\leq t \leq
T}R^{\pi^*_b}_t\geq
2^{-n}x_0]\\&=&1-\mathbb{P}^{2^{-2n}x_0}[\inf_{0\leq t \leq
T}R^{\pi^*_b}_t\geq 2^{-3n}x_0  \ \mbox{and} \sup _{0\leq t \leq
T}R^{\pi^*_b}_t\leq
2^{-n}x_0]\\
&\equiv &1-\mathbb{P}( D).
\end{eqnarray*}
Using  definition of  $a^\ast(x)$,  on the set $D$
\begin{eqnarray*} R^{\pi^*_b}_t&=&2^{-2n}x_0 \exp[[\lambda
\mu+r-\frac{1}{2}(\lambda ^2 \sigma ^2+\sigma_p^2)](t)+\sqrt{\lambda
^2 \sigma ^2+\sigma _p ^2} {W}_{t}]\\&:=&2^{-2n}x_0 \exp[X_t],
\end{eqnarray*}
where $X_t$ is a Brownian motion with drift. So
\begin{eqnarray*}
f(n)&:=&\mathbb{P}^{2^{-2n}x_0}[\inf_{0\leq t \leq
T}R^{\pi^*_b}_t\geq 2^{-3n}x_0 \ \mbox{and}\sup _{0\leq t \leq
T}R^{\pi^*_b}_t\leq 2^{-n}x_0]\\&=&\mathbb{P}[\inf_{0\leq t \leq
T}X_t \geq -n \ln 2 \ \mbox {and}\sup _{0\leq t \leq T}X_t \leq  n
\ln 2]\rightarrow 1
\end{eqnarray*}
 as $n\rightarrow \infty$.
Thus $\mathbb{P}[\tau_{b}^{\pi^*_b}\leq T]=0 $ follows from
$\mathbb{P}[\tau_{b}^{\pi^*_b}\leq T]\leq
  1-f(n)$.
\end{proof}
\vskip 10pt\noindent
{\sl  The interpretation of  Theorem
\ref{theorem42} is that when $m=0$ the company of the
 model will never go to bankruptcy. Indeed, this is an ideal model and
 does not exist in reality. Thus the assumption $m>0$ in this paper is
 reasonable and more closer to real world.}
 \setcounter{equation}{0}
\section{\bf Numerical examples}
\vskip 5pt \noindent In this section we consider
some numerical samples to demonstrate the bankrupt
 probability is a decreasing function of dividend payout
  level $b$ or initial reserve $x$ based on PDE (\ref{4.16}) below.
   The dividend payout level $b (\varepsilon, m , T)$
   decreases with $\varepsilon$, and increases with $m$,
    $\sigma^2$ and $T$ via the
    equation\\ $\psi(T,b, m, x)=\varepsilon$(see (\ref{4.16})).
\begin{example}
Let $\sigma^2=\mu=1, \sigma^2_p=2, T=1, m=1$
in PDE (\ref{4.16}) below, the figures \ref{test1} and \ref{test21m=1} of the bankrupt probability $1-\phi(T,x)$ state that solvency will improve with dividend payout level $b$ or initial reserve $x$, but the company's profit will
reduce(see Lemma \ref{lemma4} below).
\end{example}
\begin{figure}[H]
 \includegraphics[width=0.75\textwidth]{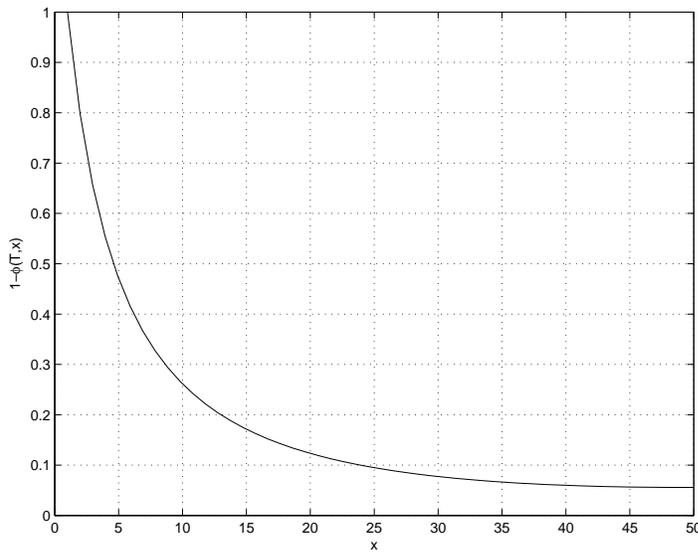}
\caption{ Bankrupt probability $1-\phi(T,x)$ as a function of $x$ (Parameters: $\sigma^2=\mu=1, \sigma^2_p=2, T=1, m=1, b=50$)}\label{test1}
\end{figure}
 \begin{figure}[H]
 \includegraphics[width=0.75\textwidth]{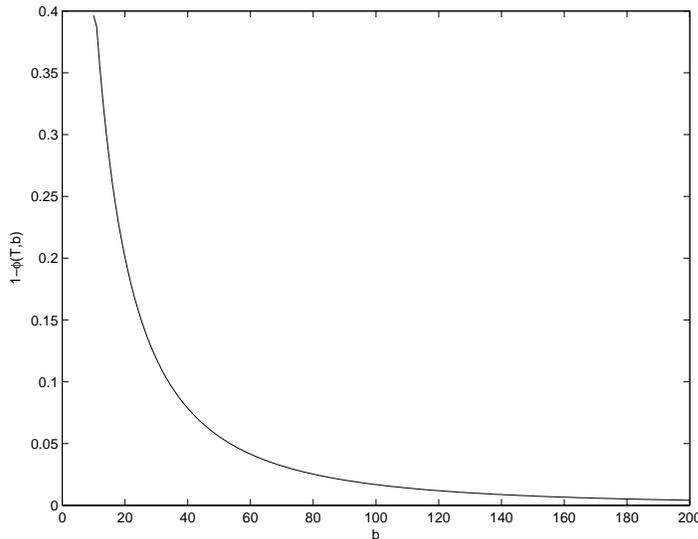}
\caption{Bankrupt probability $1-\phi(T,b)$ as a function of $b$( Parameters:  $\sigma^2=\mu=1, \sigma^2_p=2, T=1, m=1$)}\label{test21m=1}
\end{figure}
\begin{example}
Let $\sigma^2=\mu=1, \sigma^2_p=2, T=1, m=1$ and solve $b( \varepsilon )$ by $1-\phi(T,b)=\varepsilon$, we get the figure \ref{test22m=1}. It  shows that the risk $ \varepsilon$ greatly impacts on dividend payout level $b$. The dividend payout level $b$ decreases with the risk $ \varepsilon$, so the risk $ \varepsilon$   increases the company's profit.
\end{example}
\begin{figure}[H]
 \includegraphics[width=0.75\textwidth]{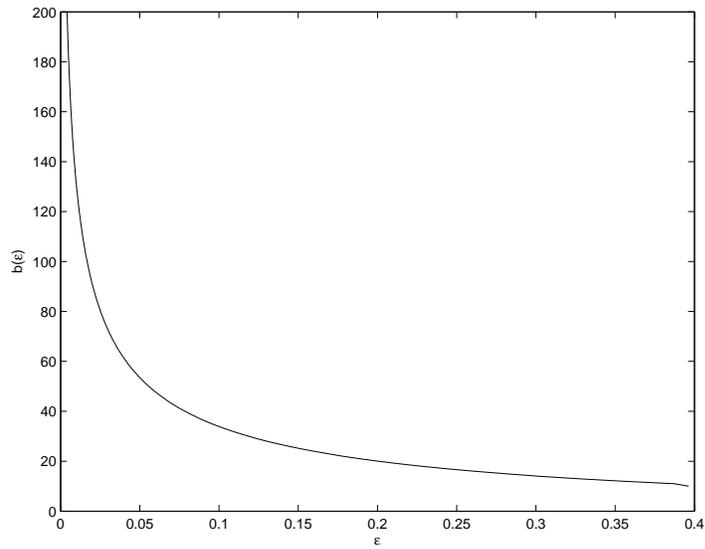}
\caption{Dividend payout level $b( \varepsilon )$  as a function of $\varepsilon$ (Parameters: $\sigma^2=\mu=1, \sigma^2_p=2, T=1, m=1$  )}\label{test22m=1}
\end{figure}
\begin{example}
Let $\sigma^2=\mu=1, \sigma^2_p=2, T=1$ and solve $b( \varepsilon )$ by $1-\phi(T,b)=\varepsilon$, we get the figure \ref{test2m=1,2} below. The two curves in this figure show that the minimum reserve requirement $m$ increases dividend payout level $b$, but decreases the company's profit.
\end{example}
\begin{figure}[H]
 \includegraphics[width=0.75\textwidth]{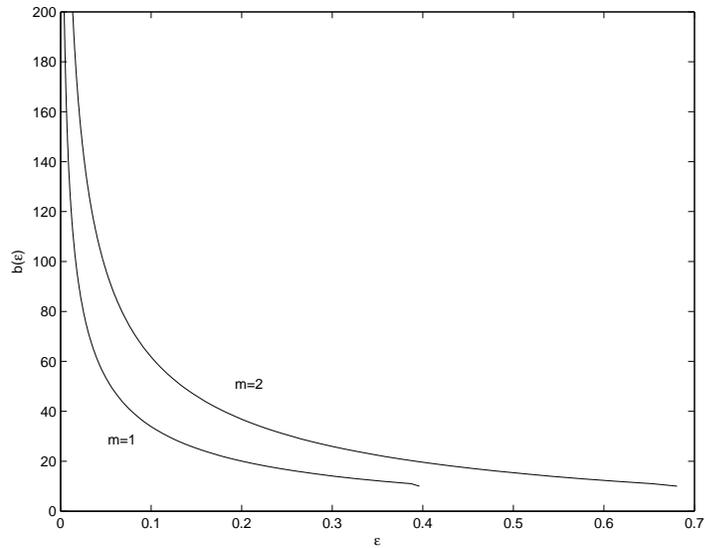}
\caption{Dividend payout level $b( \varepsilon )$
 as a function of
 $\varepsilon$ (Parameters: $ \sigma^2=\mu=1, \sigma^2_p=2, T=1$).}
 \label{test2m=1,2}
\end{figure}
\begin{example}
Let $\sigma^2=\mu=1, \sigma^2_p=2, m=1$ and solve $b( \varepsilon )$ by $1-\phi(T,b)=\varepsilon$, we get the figure \ref{test2T=1,2} below. It portrays that the dividend payout level $b$ is an increasing function of time horizon $T$, so it decreases the company's profit.
\end{example}
 \begin{figure}[H]
 \includegraphics[width=0.75\textwidth]{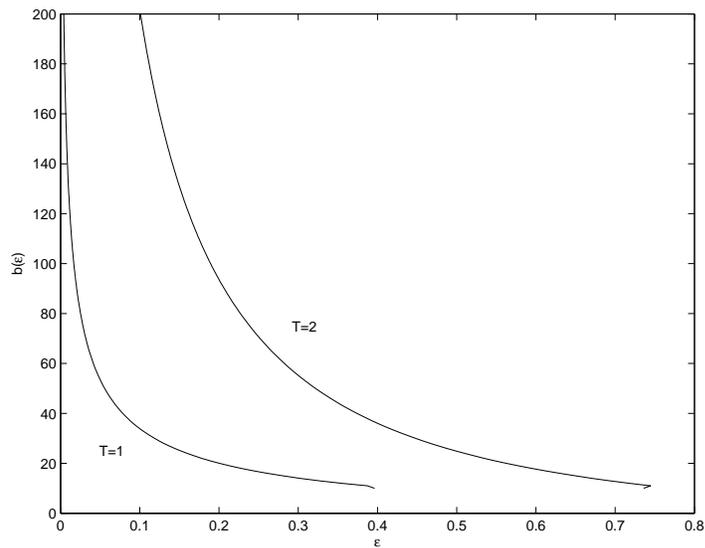}
\caption{ Dividend payout level $b( \varepsilon )$ as a function of $\varepsilon$ ( Parameters:$ \sigma^2=\mu=1, \sigma^2_p=2, m=1$).}\label{test2T=1,2}
\end{figure}
\begin{example}
Let $\mu=1, \sigma^2_p=2, m=1$ and solve $b( \varepsilon )$ by $1-\phi(T,b)=\varepsilon$, we get the figure \ref{test3sigma=1,41} below. It portrays that the dividend payout level $b$ is an increasing function of underwriting risk $\sigma^2$, so it decreases the company's profit.
\end{example}
 \begin{figure}[H]
 \includegraphics[width=0.75\textwidth]{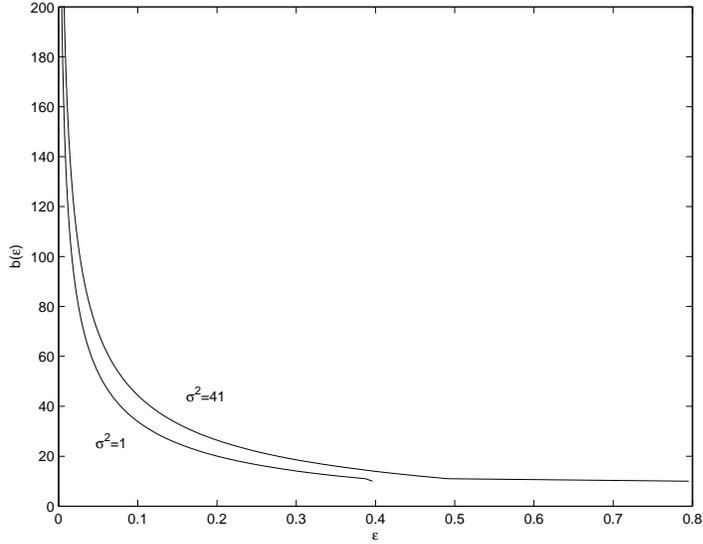}
\caption{ Dividend payout level $b( \varepsilon )$ as a function of $\varepsilon$ ( Parameters: $ \mu=1, \sigma^2_p=2, m=1$).}\label{test3sigma=1,41}
\end{figure}
 \setcounter{equation}{0}
\section{\bf Properties on bankrupt  probability  and $V(x,b)$}
\vskip 10pt\noindent
In this section, to prove Theorem 3.1, we list some lemmas on properties of  bankrupt  probability and $V(x,b)$ which will be used late. The rigorous proofs of these lemmas will be given in the appendix below.
\begin{Lemma}\label{lemma4.1} The  probability of bankruptcy
$\mathbb{P}[\tau_{b}^b \leq T]$ is a decreasing function of  $b $,
where $\tau_{b}^b:=\tau_{b}^{\pi^*_b}$.
\end{Lemma}
\begin{Lemma}\label{lemma4.2}
\begin{eqnarray}\label{4.9}
\lim_{b\rightarrow  \infty}\mathbb{P}[\tau_{b}^b \leq
T]=0.\end{eqnarray}
\end{Lemma}
\begin{Lemma}\label{lemma4.3}
Let $\phi(t,x)\in C^1(0,\infty)\cap C^2(m,b)$ and satisfy the
following partial differential equation
\begin{eqnarray}\label{4.16}
\left\{
\begin{array}{l l l}
\phi_{t}(t,y)=\frac{1}{2}[{a^\ast} ^2(x)\sigma^2+\sigma_p ^2
x^2]\phi_{xx}(
t,x)+[a^\ast(x) \mu+rx]\phi_{x}(t,x),\\
\phi(0,x)=1,\  \mbox{for}\ \  m<x\leq b, \\
\phi(t,m)=0,\phi_{x}(t,b)=0,\ \mbox{for} \ t>0.
\end{array}\right.
\end{eqnarray}
Then $\phi(T,x)=1-\psi^{b}(T,x)$, i.e., $\phi^{b}(T,x)$ is
probability that the company will survive on time interval $[0,T]$, the function $ \psi^b(t,x)$ is defined   by
\begin{eqnarray*}
\psi^b(t,x):=\mathbb{P}[\tau_{x}^{b}\leq t],
\end{eqnarray*}
 where $\tau_{b}^x:=\tau_{x}^{\pi^*_b}$,
i.e., probability of bankruptcy for the process
$\{R_{t}^{\pi^*_b,x}\}_{t\geq 0} $  with the initial asset  $x$  and
a dividend barrier $b$ is employed before time $t$.
where $ a^\ast(\cdot)$ is defined by (\ref{33}).
\end{Lemma}

\vskip 5pt\noindent Let
$\sigma(x):=\frac{1}{2}[{a^\ast}^2(x)\sigma^2+\sigma_p ^2 x^2]$ and
$\mu(x):=a^\ast(x) \mu+rx $. Then the equation (\ref{4.16}) becomes
\begin{eqnarray}\label{eq4.11}
\phi_{t}(t,x)=\sigma^2(x)\phi_{xx}( t,x)+\mu(x)\phi_{x}(t,x).
\end{eqnarray}
By  properties of $a^\ast (\cdot)$, it is  easy to show that
$\sigma(x)$ and $\mu(x)$ are continuous in $[m,b]$. So there exists
a unique solution  (\ref{4.16}) and the solution is in
$C^1(0,\infty)\cap C^2(m,b)$. Moreover, $\sigma^{'}(x)$ and
$\mu^{'}(x)$ are bounded on $(m,x_0)$ and $(x_0,b)$ respectively.
\begin{Lemma}\label{lemma4.4}
Let $\phi^{b}(t,x)$ be a solution of  the equation(\ref{4.16}). Then
the $\phi^{b}(T,b)$ is a continuous  function of  $b$ on $[ b_0,
\infty)$.
\end{Lemma}
\begin{Lemma}\label{lemma2}
Let $F_b(x)$ be defined by (\ref{32}) and $b_0$ be given  by part
(ii) of Lemma \ref{lemma1}. Then
\begin{eqnarray}\label{fl14}
\mathcal {L} F_b(x)\leq 0,\ \mbox{ for all   $x\geq 0$},
\end{eqnarray}
where
\begin{eqnarray*}
\mathcal {L} =\frac{1}{2}(a^2 \sigma ^2 +\sigma_p^2
x^2)\frac{d^{2}}{dx^{2}}+(a\mu +rx)\frac{d}{dx}-c.
\end{eqnarray*}
\end{Lemma}
\vskip 10pt\noindent
\begin{Lemma}\label{lemma3}
(i) For any  $b\leq b_0$ we have
$V(x,b)=V(x,b_0)=V(x)=F_{b_0}(x)=J(x,\pi_{b_o}^\ast)   $. Moreover,
the optimal  policy is
$\pi_{b_o}^\ast=\{a^\ast(R^{\pi_{b_o}^\ast}_t),L^{\pi_{b_o}^\ast}_t\}$,
where $ ( R^{\pi_{b_o}^\ast}_t, L^{\pi_{b_o}^\ast}_t       )  $
is uniquely determined by the  SDE (\ref{34}). \\
 (ii) For any  $b\geq
b_0$ we have $V(x,b)=F_b(x)=J(x,\pi_{b}^\ast)  $. The optimal
 policy
 $\pi_{b}^\ast=\{a^\ast(R^{\pi_{b}^\ast}_t),L^{\pi_{b}^\ast}_t\}$,
 where $ ( R^{\pi_{b}^\ast}_t, L^{\pi_{b}^\ast}_t    ) $
is uniquely determined by the SDE(\ref{37}).
\end{Lemma}
\vskip 10pt\noindent The lemma \ref{lemma3} mainly deals with
relationships among  $F_b(x)$, $V(x,b)$ and $V(x)$  defined by
(\ref{f3}).
\begin{Lemma}\label{lemma4}
For  any $b\geq b_0$ and $x\geq m $,
\begin{eqnarray} \label{flemma4}
\frac{d}{d b}V(x,b)<0.
\end{eqnarray}
Moreover, if $ b_1, b_2\geq b_0$ and $ x \leq \min\{b_1,b_2\}$, then
\begin{eqnarray}\label{flemma41}
\frac{V(x,b_1)}{V(x,b_2)}=\frac{h'(b_2)}{h'(b_1)}.
\end{eqnarray}
\end{Lemma}
\vskip 10pt\noindent
\setcounter{equation}{0}
\section{\bf Appendix  }
\vskip 10pt\noindent In this section we will give the proofs of theorem
 and lemmas we concerned with throughout this paper. \vskip 10pt\noindent
{\bf Proof of theorem \ref{theorem31}}. \
If $\mathbb{P}[\tau_{b_0}^{b_0}\leq T]\leq\varepsilon$, then the
conclusion is obvious because it is just the optimal control problem
 without constraints. \vskip 10pt\noindent
 Assume that  $
\mathbb{P}[\tau _{b_0} ^{b_0}\leq T]>\varepsilon $. By  Lemma
\ref{lemma4.1} and Lemma \ref{lemma4.2}, there exists a unique
$b^\ast (\geq b_0)$ such that
 \begin{eqnarray}\label{eq4.30}
 &&b^\ast=\min\{b:  \mathbb{P}[\tau _{b}^{b}\leq
T]=\varepsilon\}= \min\{ b: b \in \mathfrak{B}\},\\
 &&\mathbb{P}[\tau _b
^b\leq T]> \varepsilon,  \quad \forall b\leq b^\ast,\nonumber\\
&&\mathbb{P}[\tau _b ^b\leq T]\leq \varepsilon,  \quad \forall b\geq
b^\ast.\nonumber
\end{eqnarray}
By Lemma \ref{lemma4}, we know that $V(x,b)$ is decreasing w.r.t. $
b$, so $b^\ast$ satisfies(\ref{eq4.27}). Using Lemma \ref{lemma4.4},
we get $b^\ast\in \mathfrak{B}$ and $ \mathbb{P}[\tau _{b^\ast}
^{b^\ast}\leq T]=\varepsilon $.
 Moreover, by Lemma \ref{lemma3} and
(\ref{eq4.30}), we have
\begin{eqnarray*}\label{eq4.31}
 F_{b^\ast}(x)=V(x,b^\ast)=J(x,\pi_{b^*}^\ast)=V(x).
 \end{eqnarray*}  So the optimal policy associated with
  the optimal return function$
V(x)$ is $\{a^\ast(R^{\pi_{b^*}^\ast}_t),L^{\pi_{b^*}^\ast}_t\}$,
 where, $  ( R^{\pi_{b^*}^\ast}_t,L^{\pi_{b^*}^\ast}_t     )  $ is
 determined uniquely by (\ref{37}). The inequality (\ref{38})
 is a direct consequence of (\ref{flemma41}).  $ \Box$
 \vskip 10pt\noindent
{\bf Proof of lemma \ref{lemma4.1}}.  \  The proof of this lemma is the same as that of Theorem 3.1 in
 the \cite{HHL}, we omit it here.  $ \Box$
\vskip 10pt\noindent
 {\bf Proof of lemma \ref{lemma4.2}}. \
Using the same argument as in the  proof of theorem 3.1 in the
\cite{HHL},
 we have for some $n>3$ and   large  $ b \geq
\max\{1,m^n \}$
\begin{eqnarray}\label{4.10}
\mathbb{P}[\tau^{b}_{\sqrt[n]{b} } \leq T ] \geq \mathbb{P}[\tau_b^b
\leq T].
\end{eqnarray}
Let  $R_t^{(2)}$ be the unique solution of the following SDE
\begin{eqnarray}\label{4.11}
&&dR_t^{(2)}=(a^\ast(R_t^{(2)})\mu+rR_t^{(2)})dt+
\sqrt{{a^\ast}^2(R_t^{(2)})\sigma^2+\sigma_p^2
{R_t^{(2)}}^2}dW_t,\nonumber
\\
&&R_0^{(2)}=\sqrt[n]{b}.
\end{eqnarray}
Then by comparison theorem on SDE ( see Ikeda and Watanabe
\cite{s639}(1981))
$$\{ \tau_{\sqrt[n]{b}}^{\pi^*_b}\leq  T\}\subseteq\{ \exists\  t\leq T
\ \mbox{such that}\ R_t^{(2)}=m\ \mbox{ or}\ R_t^{(2)}=b\}.$$ As a
result,
\begin{eqnarray}\label{4.12}
\mathbb{P}\{ \tau_{ \sqrt[n]{b}   }^{\pi^*_b}\leq  T\}&\leq&
\mathbb{P}\{\exists\ t\leq T \ \mbox{such
that}\ R_t^{(2)}=m\ \mbox{ or}\ R_t^{(2)}=b\}\nonumber\\
&\leq& \mathbb{P}\{\sup_{0\leq t\leq T}R_t^{(2)}\geq
b\}+\mathbb{P}\{\inf_{0\leq t\leq T}R_t^{(2)}\leq m\}.
\end{eqnarray}
\vskip 10pt\noindent Firstly, we estimate $ \mathbb{P}\{\sup_{0\leq
t\leq T}R_t^{(2)}\geq b\}$.\vskip 10pt\noindent
 Using H\"{o}lder
inequality and $a^*(x)\leq 1$, it follows from SDE (\ref{4.11}) that
\begin{eqnarray}\label{c18}
\sup_{0\leq t\leq T}(R_t^{(2)})^2 &\leq & [ 3(  \sqrt[n]{b})^2  +
6\mu^2T^2] + 6r^2T \int^T_0\sup_{0\leq s\leq t}(R_s^{(2)})^2ds \nonumber\\
&+&3 \sup_{0\leq t\leq T}\big (  \int^t_0
\sqrt{{a^\ast}^2(R_s^{(2)}\big )\sigma^2+\sigma_p^2
{R_s^{(2)}}^2}dW_s)^2.
\end{eqnarray}
Taking mathematical expectation at both sides of (\ref{c18}) and
using B-D-G inequality, we derive
\begin{eqnarray}\label{c57}
\mathbb{E}\{\sup_{0\leq t\leq T}(R_t^{(2)})^2\} & \leq & [ 3(
\sqrt[n]{b})^2  + 6\mu^2T^2] + 6r^2T \int^T_0\mathbb{E}\{\sup_{0\leq
s\leq t}(R_s^{(2)})^2\}ds\nonumber \\
&&+12 \mathbb{E}\{\int^T_0 ({a^\ast}^2(R_s^{(2)})\sigma^2+\sigma_p^2
{R_s^{(2)}}^2)  dt\}\nonumber\\
&\leq &  \big  [ 3 (\sqrt[n]{b})^2  + 6\mu^2T^2 +12\sigma^2 T^2 \big
]\nonumber\\
&&+ 6( r^2T + 2\sigma^2_p ) \int^T_0\mathbb{E}\{\sup_{0\leq s\leq
t}(R_s^{(2)})^2\}ds.
\end{eqnarray}
Solving (\ref{c57}), we get
\begin{eqnarray}\label{4.13}
\mathbb{E}\{\sup_{0\leq t\leq T}(R_t^{(2)})^2\} & \leq & \big [( 3
\sqrt[n]{b})^2  + 6\mu^2T^2 +12\sigma^2 T^2 \big ]\exp\{6( r^2T +
2\sigma^2_p )T    \}.\nonumber\\
 \end{eqnarray}
 Combining  Markov inequality and the inequality (\ref{4.13}),  we conclude that
 \begin{eqnarray}\label{4.14}
\mathbb{P}\{\sup_{0\leq t\leq T}R_t^{(2)}\geq b\}&\leq & \frac{
  \mathbb{E}\{\sup_{0\leq t\leq T}(R_t^{(2)})^2\}            }{b^2}
\nonumber\\
  & \leq & \frac{  \big( 3 (\sqrt[n]{b})^2  + 6\mu^2T^2 +12\sigma^2 T^2 \big
)\exp\{6( r^2T + 2\sigma^2_p )T    \}      }{  b^2   }.\nonumber\\
 \end{eqnarray}
\vskip 10pt\noindent Secondly, we estimate $\mathbb{P}\{\inf_{0\leq
t\leq T} R_t^{(2)}\leq m\}$. \vskip 10pt\noindent
 Let  $ M_1 $ be a  martingale  defined  by
$$ M_1(t)=\int^t_0I_{\{s:   R_s^{(2)}\neq 0\}}
 \sqrt{\sigma^2  ( \frac{ a^\ast(R_s^{(2)}) }{R_s^{(2)}}    )^2 +
\sigma_p^2     } \quad dW_s.$$ Then we can rewrite the SDE
(\ref{4.11}) as follows,
$$ R_t^{(2)}=\big ( \sqrt[n]{b} +
\int^t_0(a^\ast(R_s^{(2)})\mu+rR_s^{(2)})ds\big )
+\int^t_sR_s^{(2)}dM_1(s).$$
 In view of  Proposition 2.3 of Chapter 9 in
\cite{DRMY},
$$  R_t^{(2)}=\mathcal{E}(M_1)_t \big ( \sqrt[n]{b}  +\int^t_0
\frac{(\mu a^\ast(R_s^{(2)})+rR_s^{(2)} )}{\mathcal{E}(M_1)_s   }
ds\big ),
$$
where $ \mathcal{E}(M_1)_t= \exp\{ M_1(t)-\frac{1}{2}<M_1>(t)\}  $
is an exponential martingale,  $<M_1>$ is the  bracket of $M_1$. So
the fact $\inf{\{f(t)g(t)\}}\geq \inf\{f(t)\}\inf\{g(t)\} $ for any
$f(t)\geq 0 $ and $g(t)\geq 0$ implies that
$$ \inf_{0\leq t \leq T}\{  R_t^{(2)}  \}
\geq \sqrt[n]{b} \inf_{0\leq t \leq T}\{\mathcal{E}(M_1)_t \}.$$ As
a result
$$ \mathbb{P}\{\inf_{0\leq t\leq T} R_t^{(2)}\leq m\}\leq \mathbb{P}\{\inf_{0\leq
t\leq T}\mathcal{E}(M_1)_{t}\leq \frac{m}{\sqrt[n]{b}}\} .      $$
Since $<M_1>_T\leq ( \lambda^2\sigma^2+ \sigma_p^2)T < + \infty$, we
have
\begin{eqnarray}\label{c20}
\lim_{b\longrightarrow \infty }\mathbb{P}\{\inf_{0\leq t\leq T}
R_t^{(2)}\leq m\}&\leq &\mathbb{P}\{\inf_{0\leq t\leq
T}\mathcal{E}(M_1)_{t}=0\}\nonumber\\
&\leq &\mathbb{P}\{\sup_{0\leq t\leq T}|M_1(t)|=+\infty\} .
\end{eqnarray}
By B-D-G inequalities, we get
$$\mathbb{E}\{ \sup_{0\leq t\leq T}|M_1(t)|^2 \}\leq
4 ( \lambda^2\sigma^2+ \sigma_p^2)T < +  \infty , $$ which implies
that
$$  \mathbb{P}\{\sup_{0\leq t\leq T}|M_1(t)|=+\infty\}=0. $$
Thus by (\ref{c20})
\begin{eqnarray}\label{4.15}
\lim_{b\longrightarrow \infty }\mathbb{P}\{\inf_{0\leq t\leq T}
\{R_t^{(2)}\}\leq m\}=0.
\end{eqnarray}
So  the inequalities  (\ref{4.10}), (\ref{4.12}), (\ref{4.14}) and
(\ref{4.15}) yield that
$$\lim_{b\rightarrow \infty}\mathbb{P}\{ \tau_{{b}}^{\pi^*_b}\leq  T\}=0.$$
 $\Box$
\vskip 10pt\noindent
\begin{Remark}
The proof of theorem 3.2 in \cite{HHL} seems  wrong, so we can use
the way proving Lemma \ref {lemma4.2}  to correct it. Theorem 3.2 in
the \cite{HHL} is indeed a direct consequence of Lemma \ref
{lemma4.2}.
\end{Remark}\vskip 10pt\noindent
{\bf Proof of lemma \ref{lemma4.3}}.\
Let $( R_{t}^{b},L_b(t) ) $ denote $ (  R^{\pi_{b}^\ast}_t,
L^{\pi_{b}^\ast}_t  )   $ defined by  SDE (\ref{c21}). Since $
(R_{t}^{b},L_b(t) )$ is continuous process, by the generalized
It\^{o} formula,  we have
\begin{eqnarray}\label{4.17}
\phi(T-(t\wedge\tau_{x}^{b}),R_{t\wedge\tau_{x}^{b}}^{b})
&=&\phi(T,x)\nonumber\\
&+&\int_{0}^{t\wedge\tau_{x}^{b}}\{\frac{1}{2}[{a^\ast}^2(R_s
^b)\sigma^2+\sigma_p^2 \cdot (R_s ^b)^2]
\sigma^{2}\phi_{xx}(T-s,R_{s}^{b})\nonumber\\
&+&[ a^{*}(R_{s}^{b})\mu+rR_{s}^{b}]
\phi_{x}(T-s,R_{s}^{b})\nonumber\\
&-&\phi_{t}(T-s, R_{s}^{b})\}ds-\int_{0}^{t\wedge\tau_{x}^{b}}\phi_{x}(T-s,R_{s}^{b})dL_b(s)\nonumber\\
&+&\int_{0}^{t\wedge\tau_{x}^{b}}
a(R_{s}^{b})\sigma\phi_{x}(T-s,R_{s}^{b})dW_{s}.
\end{eqnarray}
 Letting  $t=T$
and taking mathematical expectation at both sides of (\ref{4.17})
yields that
\begin{eqnarray*}
\phi(T,x)&=&\mathbb{E}[\phi(T-(T\wedge\tau_{x}^{b}),R_{T\wedge\tau_{x}^{b}}^{b})]\nonumber\\
&=&\mathbb{E}[\phi(0,R_{T}^{b})1_{T<\tau_{x}^{b}}]+\mathbb{E}
[\phi(T-\tau_{x}^{b},m)1_{T\geq\tau_{x}^{b}})]\nonumber\\
&=&\mathbb{E}[1_{T<\tau_{x}^{b}}]=\mathbb{P}[\tau_{x}^{b}>T]=1-\psi(T,x).
\end{eqnarray*} $\Box$
\vskip 10pt\noindent
Now we use  PDE method to prove lemma \ref{lemma4.4}.
\vskip 10pt\noindent
{\bf Proof of lemma \ref{lemma4.4}}. \  Let  $x=by,0\leq y \leq 1$ and
$\theta^{b}(t,y)=\phi^{b}(t,(b-m)y+m)$. Then the equation
(\ref{4.16}) becomes
\begin{eqnarray}\label{eq4.12}
\left\{
\begin{array}{l l l}
\theta_t^{b}(t,y)=[\sigma[(b-m)y+m]/(b-m)^2]\theta_{yy}^{b}(t,y)\\+[\mu[(b-m)y+m]/(b-m)]\theta_y^{b}(t,y),\\
\theta^{b}(0,y)=1, \ \mbox{for}\  0\leq y\leq 1, \\
\theta^{b}(t,0)=0,\theta_{y}^{b}(t,1)=0,\ \mbox{for} \ t>0.
\end{array}\right.
\end{eqnarray}
In view of (\ref{eq4.12}), the proof of Lemma \ref{lemma4.4} reduces
to proving $\lim\limits_{b_2\rightarrow
b_1}\theta^{b_2}(t,1)=\theta^{b_1}(t,1)$ for fixed $b_1>b_0$. Let
$w(t,y)=\theta^{b_2}(t,y)-\theta^{b_1}(t,y)$. Since
$\theta^{b}(t,y)$ is continuous at $y=1$ for any $b>b_0$, we only
need to show that
 \begin{eqnarray}\label{eq4.13}
 \int_0^t\int_0^1 w^2(s,y)dsdy
\rightarrow 0, \mbox {as $ b_2\rightarrow b_1$}.
\end{eqnarray}
Let $\sigma^b(y)=\sigma[(b-m)y+m]/(b-m)^2$,
 $\mu^b(y)=\mu[(b-m)y+m]/(b-m)$. Then the (\ref{eq4.12}) translates into
 \begin{eqnarray}\label{eq4.15}
\left\{
\begin{array}{l l l}
w_t(t,y)&=&\sigma^{b_2}(y)w_{yy}(t,y)+\mu^{b_2}(y)w_{y}(t,y)\\
&+& \{\sigma^{b_2}(y)-\sigma^{b_1}(y)\}\theta_{yy}^{b_1}(t,y)\\
&+&\{\sigma^{b_2}(y)-\sigma^{b_1}(y)\}\theta_y^{b_1}(t,y),\\
w(0,y)&=&0,\ \mbox{for}\ 0<y\leq 1,\\
w(t,y)&=&0,y=0, \ w_y(t,1)=0,\ \mbox{for}\ t>0.
\end{array}\right.
\end{eqnarray}
 Multiplying both sides of the first equation in
(\ref{eq4.15}) by $w(t,z)$, and then  integrating both sides of the
resulting equation on $[0,t]\times [0,1]$, we get
\begin{eqnarray}\label{eq4.16}
&&\int_0^{t}\int_0^1 w(s,y)w_t(s,y)dyds
\nonumber\\
&=&\int_0^{t}\int_0^1 [\sigma^{b_2}(y)]
w(s,y)w_{yy}(s,y)dyds\nonumber\\
&+&\int_0^{t}\int_0^1[\mu^{b_2}(y)]w(s,y)w_y(s,y)dyds\nonumber\\
 &+&\int_0^{t}\int_0^1[\sigma^{b_2}(y)-\sigma^{b_1}(y)]w(s,y)\theta_{yy}^{b_1}(t,y)dyds\nonumber\\
  &+&\int_0^{t}\int_0^1w(s,y)[\mu^{b_2}(y)-\mu^{b_1}(y)]w(s,y)\theta_y^{b_1}(t,y)dyds\nonumber\\
&\equiv & E_1 +E_2+E_3+E_4.
\end{eqnarray}
Now we look at terms at both sides of (\ref{eq4.16}). \vskip
10pt\noindent Firstly, we have
\begin{eqnarray}\label{eq4.17}
\int_0^{t}\int_0^1 w(s,y)w_t(s,y)dyds&=&\int_0^1
\frac{1}{2}w^2(t,y)dy.
\end{eqnarray}
Secondly, we deal with terms $E_i$, $i=1,\cdots , 4$ as follows.
\vskip 10pt\noindent
 It is easy
to see from the expression of $a^\ast(\cdot)$ that there exist
positive constants $D_1$, $D_2$ and $D_3$ such that
$[\mu(b_2y)/b_2]^2\leq D_1$ and $\sigma^b_2(y)\geq D_2>0$ for $y\geq
0$, and
 $\sigma^b_2(y)'\leq D_3$
 for $y\in (0,\frac{x_0-m}{b-m})\cup
(\frac{x_0-m}{b-m},1]$. As a result, for any $\lambda_1 >0 $ and
$\lambda_2 >0$
\begin{eqnarray}\label{eq4.18}
E_1&=&\int_0^{t}\int_0^1[\sigma^{b_2}(y)]w(s,y)w_{yy}(s,y)dyds\nonumber\\
&=&-\int_0^{t}\int_0^1 [\sigma^{b_2}(y)]w_y^2(s,y)dyds\nonumber\\
&&-
\int_0^{t}[\int_0^{(x_0-m)/(b-m)}+
\int_{(x_0-m)/(b-m)} ^1][\sigma^{b_2}(y)]^{'}w_y(s,y)w(s,y)dyds\nonumber\\
&\leq & -D_2\int_0^{t}\int_{0}^1
w_y^2(s,y)dyds\nonumber\\
&& +D_3\int_0^{t}\int_0^1 \lambda_1 w_y^2(s,y)+\frac{1}{4\lambda_1}
w^2(s,y)dyds
\end{eqnarray}
and
\begin{eqnarray}\label{eq4.19}
E_2&=&\int_0^{t}\int_0^1[\mu^{b_2}(y)]w(s,y)w_y(s,y)\nonumber\\
&\leq
&\lambda_2\int_0^{t}\int_0^1 w_y^2(s,y)dyds\nonumber\\
&&+\frac{D_1}{4\lambda_2}\int_0^{t}\int_0^1 w^2(s,y)dyds.
\end{eqnarray}
In order to estimate  $E_3$, we decompose $E_3 $ as follows:
\begin{eqnarray}\label{eq4.20}
E_3&=&\int_0^t \int_0^1 \{\sigma^{b_2}(y)-\sigma^{b_1}(y)\}w(s,y)\theta_{yy}^{b_1}(s,y)dyds\nonumber\\
&=&-\int_0^t \int_0^1 \{\sigma^{b_2}(y)-\sigma^{b_1}(y)\}w_y(s,y)\theta_{y}^{b_1}(s,y)dyds\nonumber\\
&&-\int_0^t \int_{0}^{(x_0-m)/(b_2-m)}\{\sigma^{b_2}(y)-\sigma^{b_1}(y)\}'w(s,y)\theta_{y}^{b_1}(s,y)dyds\nonumber\\
&&-\int_0^t \int_{(x_0-m)/(b_2-m)}^{(x_0-m)/(b_1-m)}\{\sigma^{b_2}(y)-\sigma^{b_1}(y)\}'w(s,y)\theta_{y}^{b_1}(s,y)dyds\nonumber\\
&&-\int_0^t \int_{(x_0-m)/(b_1-m)}^{1}\{\sigma^{b_2}(y)-\sigma^{b_1}(y)\}'w(s,y)\theta_{y}^{b_1}(s,y)dyds\nonumber\\
&=& E_{31}+E_{32}+E_{33}+E_{34}.
\end{eqnarray}
So the estimating $E_3$ is reduced to estimating  $E_{3i}$,
$i=1,\cdots, 4$. \vskip 10pt\noindent
The fact $[\sigma(by)/b^2]$,
$[\sigma(by)/b^2]'$and $[\mu(by)/b]$ are Lipschitz continuous on
$(0,\frac{x_0-m}{b_2-m})$,
$(\frac{x_0-m}{b_2-m},\frac{x_0-m}{b_1-m})$ and
$(\frac{x_0-m}{b_1-m},1)$, that is,  there exists  $L>0$ such that
\begin{eqnarray*}
&&|[\sigma^{b_2}(y)]-[\sigma^{b_1}(y)]|\leq L|b_2-b_1|,\\
&&|[\sigma^{b_2}(y)]'-[\sigma^{b_1}(y)]'|\leq L|b_2-b_1|,\\
&& |[\mu^{b_2}(y)]-[\mu^{b_1}(y)]|\leq L|b_2-b_1|,
\end{eqnarray*}
and Young's inequality yield that for any $\lambda_3>0$ and
$\lambda_4
>0$
\begin{eqnarray}\label{eq4.21}
 E_{31}&=& -\int_0^t\int_0^1
[\sigma^b_2(y)-\sigma^b_1(y)]w_y(s,y)\theta_{y}^{b_1}(s,y)dyds\nonumber\\
&\leq &\frac{L^2(b_2-b_1)^2}{4\lambda_3}\int_0^t\int_0^1
[\theta_{y}^{b_1}(s,y)]^2dyds\nonumber\\
&& +\lambda_3\int_0^t\int_0^1 w_y^2(s,y)dyds,
\end{eqnarray}
\begin{eqnarray}\label{eq4.22}
E_{32}+E_{34}&=&
-\int_0^t\big[\int_{0}^{(x_0-m)/(b_2-m)}+\int_{(x_0-m)/(b_1-m)}^{1}\big]
\{\sigma^{b_2}(y)\nonumber\\&-&\sigma^{b_1}(y)\}'w(s,y)\theta_{y}^{b_1}(s,y)dyds\nonumber\\
&\leq &\frac{L^2(b_2-b_1)^2}{4\lambda_4}\int_0^t\int_0^1
[\theta_{y}^{b_1}(s,y)]^2dyds\nonumber\\
&& +\lambda_4\int_0^t\int_0^1 w^2(s,y)dyds.
\end{eqnarray}
The remaining part of estimating $E_3$ is to deal with
$E_{33}$.\vskip 10pt\noindent
By the boundary conditions
\begin{eqnarray*}
0&=&\int_0^t\int_0^1 \theta_{t}^{b}(s,y)\theta^{b}(s,y)\\
&& -\sigma^b(y)\theta_{yy}^{b}(s,y)\theta^{b}(s,y)
-\mu^b(y)\theta_{y}^{b}(s,y)\theta^{b}(s,y)dyds\\
&=&\frac{1}{2}\int_0^1 [\theta^{b}(s,y)]^2dy
+\int_0^t\int_0^1 \sigma^b(y) [\theta_{y}^{b}(s,y)]^2dyds\\
&&+\int_0^t\big[\int_0^{(x_0-m)/(b-m)}+\int_{(x_0-m)/(b-m)}^1\big]
\sigma^b(y)'[\theta_{y}^{b}(s,y)][\theta^{b}(s,y)]dyds\\
&=&\frac{1}{2}\int_0^1 [\theta^{b}(s,y)]^2dy
+\int_0^t\int_0^1 \sigma^b(y) [\theta_{y}^{b}(s,y)]^2dyds\\
&&+\int_0^t\big[\int_0^{(x_0-m)/(b-m)}+\int_{(x_0-m)/(b-m)}^1\big]
(\sigma^b(y)'\\
&&-\mu^b(y))[\theta_{y}^{b}(s,y)][\theta^{b}(s,y)]dyds\\
&\geq & \lambda_5 \int_0^t\int_0^1 [\theta_{y}^{b}(s,y)]^2dyds
-\frac{\lambda_5 }{2}\int_0^t\int_0^1 [\theta_{y}^{b}(s,y)]^2dyds\\
&&-\frac{\lambda_6}{2\lambda_5}\int_0^t\int_0^1 [\theta^{b}(s,y)]^2dyds\\
&\geq &\frac{\lambda_5}{2}\int_0^t\int_0^1
[\theta_{y}^{b}(s,y)]^2dyds -\frac{\lambda_6}{2\lambda_5},
\end{eqnarray*}
from which we know that
\begin{eqnarray*}
\int_0^t\int_0^1 [\theta_{y}^{b}(s,y)]^2dyds\leq
\frac{\lambda_6}{\lambda_5 ^2},
\end{eqnarray*}
where $\lambda_5>0$ is the lower boundary of $\sigma^b(y)$ and
$\lambda_6$ is the upper boundary of
$|[\sigma(by)/b^2]'-[\mu(by)/b]|$ on $(0,\frac{x_0-m}{b-m})\cup
(\frac{x_0-m}{b-m},1]$.\vskip 10pt\noindent
 Therefore we conclude
that $\int_0^t\int_0^1 [\theta_{y}^{b}(s,y)]^2dyds$ is bounded. So
by using  $w(s,y)\leq 2$, we have
\begin{eqnarray}\label{eq4.23}
\lim_{b_2\rightarrow b_1} |E_{33}|=0.
\end{eqnarray}
Thus the equalities (\ref{eq4.21}),(\ref{eq4.22}) and (\ref{eq4.23})
yield that there exists a positive function $B_1^{b_1}(b_2)$ with $
\lim\limits_{b_2\rightarrow b_1}B_1^{b_1}(b_2)=0$ such that for $
0\leq t\leq T$
\begin{eqnarray}\label{eq4.24}
E_3&=&\int_0^t\int_{0}^{1}\{\sigma(b_2y)/b_2^2-\sigma(b_1y)/b_1^2\}
w(s,y)\theta_{yy}^{b_1}(t,y)dyds\nonumber\\
&=&E_{31}+E_{32}+E_{33}+E_{34}\nonumber\\
&\leq & B^{b_1}(b_2)+(\lambda_3+\lambda_4) \int_0^t\int_0^1
w_y^2(s,y)+w^2(s,y)dyds.\nonumber\\
\end{eqnarray}
By the same way as that  of (\ref{eq4.19})
\begin{eqnarray}\label{eq4.25}
E_4&=&\int_0^t\int_0^1 \{\mu(b_2y)/b_2-\mu(b_1y)/b_1\}w(s,y)
\theta_y^{b_1}(t,y)dyds\nonumber\\
&\leq &\frac{L^2(b_2-b_1)^2}{4\lambda_7}\int_0^t\int_0^1
[\theta_{y}^{b_1}(s,y)]^2dyds\nonumber\\
&& +\lambda_7\int_0^t\int_0^1 w^2(s,y)dyds.
\end{eqnarray}
Let
\begin{eqnarray*}
B_2^{b_1}(b_2)=\frac{L^2(b_2-b_1)^2}{4\lambda_7}\int_0^t\int_0^1
[\theta_{y}^{b_1}(s,y)]^2dyds.
\end{eqnarray*}
Then
\begin{eqnarray*}
\lim\limits_{b_2\rightarrow b_1}B_2^{b_1}(b_2)=0,
\end{eqnarray*}
which, together with (\ref{eq4.25}), implies that
\begin{eqnarray}\label{eq4.26}
E_4\leq   B_2^{b_1}(b_2)+\lambda_7\int_0^t\int_0^1 w^2(s,y)dyds.
\end{eqnarray}
Choosing  $\lambda_1$, $\lambda_2$ and $\lambda_3$ small enough such
that$\lambda_1+D_3\lambda_2 +\lambda_3<D_2$, we can conclude from
$(\ref{eq4.16}) $, (\ref{eq4.18}), (\ref{eq4.19}), (\ref{eq4.24})
and (\ref{eq4.26}) that there exist constants $C_1$ and $C_2$ such
that
\begin{eqnarray*}
\int_0^1 w^2(t,y)dy \leq C_1\int_0^t \int_0^1
w^2(s,y)dyds+C_2[B_1^{b_1}(b_2)+B_2^{b_1}(b_2)].
\end{eqnarray*}
Using the Gronwall inequality, we get
\begin{eqnarray*}
\int_0^t\int_0^1 w^2(s,y)dyds\leq
C_2[B_1^{b_1}(b_2)+B^{b_1}(b_2)]\exp\{C_1t\}.
\end{eqnarray*}
So
\begin{eqnarray*}
\lim\limits_{b_2\rightarrow
b_1}\int_0^t\int_0^1[\theta^{b_2}(s,y)-\theta^{b_1}(s,y)]^2dyds=0.
\end{eqnarray*}
Thus we complete the proof.\  $\Box$
\vskip 10pt\noindent
{\bf Proof of lemma \ref{lemma2}}.\  If  $x < m $ then  by (\ref{32}), $F_b(x)=0$.
It suffices to prove (\ref{fl14}) for $m \leq x  $. If $m\leq x\leq
b $, then
\begin{eqnarray*}
\mathcal {L} F_b(x)=\frac{\mathcal {L}h(x)}{h'(b)},
\end{eqnarray*}
here $h(x)$ is a solution of (\ref{31}), so $\mathcal {L} F_b(x)
\leq 0$ follows from lemma \ref{lemma1}.  If $x>b $ then by using
$F_{b}''(b)\geq 0$ for $b\geq b_0$
\begin{eqnarray*}\label{f50}
\mathcal {L}F_{b}(x)&=&\frac{1}{2}(a^2 \sigma ^2 +\sigma_p^2
x^2)F_b''(x)+(a \mu +rx)F_b'(x)-cF_b(x)\nonumber\\
&\leq &(\mu+rx)-c(x-b+F_{b}(b))\nonumber\\
&\leq&(\mu+rb)-cF_{b}(b)\nonumber\\
 &=&\mathcal {L}F_{b}(b)-\frac{1}{2}(a^2 \sigma ^2 +\sigma _p ^2x^2)F_{b}''(b)\nonumber\\
 &\leq& 0
\end{eqnarray*}
 Thus the proof follows.\ $\Box$
 \vskip 10pt\noindent
{\bf Proof of lemma \ref{lemma3}}. \  The proof basically follows  the same arguments as in  the proof of
theorem 5.2 in He and Liang \cite{he and Liang} and so we omit it.\
$\Box$
\vskip 10pt\noindent
{\bf Proof of lemma \ref{lemma4}}.\
The lemma is a direct consequence of lemma \ref{lemma2}
and lemma \ref{lemma3}. $\Box$
\vskip 20pt\noindent
 {\bf Acknowledgements.} This work is supported
by Project 10771114 of NSFC, Project 20060003001 of SRFDP, the SRF
for ROCS, SEM  and the Korea Foundation for Advanced Studies. We
would like to thank the institutions for the generous financial
support. We are very grateful to the referees for the careful reading of the manuscript, correction of errors,  and valuable suggestions which improved the main results of this paper very much. Special thanks also go to the participants of the seminar stochastic analysis, finance and insurance at Tsinghua University for their
feedbacks and useful conversations. Zongxia Liang is also very
grateful to College of Social Sciences and College of Engineering at
Seoul National University for providing excellent working conditions
for him. The authors also thank Jicheng Yao for very valuable
discussions on lemma \ref{lemma4.4}.
\vskip 10pt\noindent
\setcounter{equation}{0}

\end{document}